\documentclass[reqno,a4paper]{amsart}
\renewcommand{\paragraph}{\subsubsection*}
\usepackage[UKenglish]{babel}
\usepackage{microtype}
\usepackage{amssymb}
\usepackage{amsmath}
\usepackage{amsthm}
\usepackage{amsfonts}
\usepackage{xcolor}
\usepackage{booktabs}
\def\figwidth{\columnwidth}
\usepackage[style=alphabetic,natbib=true,maxbibnames=99,maxalphanames=99,backend=bibtex]{biblatex}
\renewbibmacro*{volume+number+eid}{%
  \printfield{volume}%
  \printfield{number}%
  \setunit{\addcomma\space}%
  \printfield{eid}}
\DeclareFieldFormat[article]{number}{\mkbibparens{#1}}
\usepackage{thm-restate}
\newcounter{oldthmcount}
\newcounter{oldclmcount}
\newcounter{newthmcount}

\newenvironment{claimproof}[1][\proofname]
  {\begin{proof}[\proofname]}
  {\end{proof}}
\usepackage[foot]{amsaddr}
\usepackage{lineno}
\usepackage{url}
\usepackage[
  pdftex,
  bookmarks=true,
  bookmarksnumbered=true,
  hypertexnames=false,
  breaklinks=true,
  linkbordercolor={0 0 1}
  ]{hyperref}

\providecommand{\urlstyle}[1]{}
\providecommand{\doi}[1]{\href{http://dx.doi.org/#1}{\nolinkurl{doi:#1}}}
\usepackage{cleveref}

\usepackage{tikz}
\usetikzlibrary{patterns}
\usetikzlibrary{calc}
\usetikzlibrary{arrows.meta}
\usetikzlibrary{positioning}
\usetikzlibrary{decorations.pathmorphing}
\usepackage[shortlabels]{enumitem}  

\usepackage{todonotes}

\usepackage{thmtools}
\usepackage{thm-restate}
\usepackage{xpatch}
\makeatletter
\xpatchcmd \thmt@restatable{\begingroup}
                           {\begingroup\def\thmt@tmp@restatename{#3}}{}{\fail}
\newtheoremstyle{claim}%
{}{}%
{\normalfont}{}%
{\itshape}{.}{ }%
{\ifcsname thmt@tmp@restatename\endcsname 
   \ifthmt@thisistheone%
      {\thmname{#1} \hyperlink{proof\thmt@tmp@restatename}{\thmnumber{#2}}\thmnote{ (#3)}}%
    \else%
      {\thmname{#1} \hypertarget{proof\thmt@tmp@restatename}{\thmnumber{#2}}\thmnote{ (#3)}}%
   \fi%
 \else%
   \thmname{#1} \thmnumber{#2}\thmnote{ (#3)}%
 \fi}
\newtheoremstyle{announce}%
{}{}%
{\itshape}{}%
{}{{\bfseries.}}{ }%
{\ifcsname thmt@tmp@restatename\endcsname 
   \ifthmt@thisistheone%
      {{\bfseries\thmname{#1} \hypertarget{proof\thmt@tmp@restatename}{\thmnumber{#2}}}\thmnote{ (#3)}}%
    \else%
      {{\bfseries\thmname{#1} \hyperlink{proof\thmt@tmp@restatename}{\thmnumber{#2}}}\thmnote{ (#3)}}%
   \fi%
 \else%
   {\bfseries\thmname{#1} \thmnumber{#2}}\thmnote{ (#3)}%
 \fi}
\newtheoremstyle{restatable}%
{}{}%
{\itshape}{}%
{}{{\bfseries.}}{ }%
{\ifcsname thmt@tmp@restatename\endcsname 
   \ifthmt@thisistheone%
      {{\bfseries\thmname{#1} \hyperlink{proof\thmt@tmp@restatename}{\thmnumber{#2}}}\thmnote{ (#3)}}%
    \else%
      {{\bfseries\thmname{#1} \hypertarget{proof\thmt@tmp@restatename}{\thmnumber{#2}}}\thmnote{ (#3)}}%
   \fi%
 \else%
   {\bfseries\thmname{#1} \thmnumber{#2}}\thmnote{ (#3)}%
 \fi}
\makeatother
\theoremstyle{announce}
\newtheorem{theorem}{Theorem}
\newtheorem{lemma}[theorem]{Lemma}
\newtheorem{proposition}[theorem]{Proposition}
\newtheorem{corollary}[theorem]{Corollary}
\theoremstyle{restatable}

\theoremstyle{remark}

\theoremstyle{definition}
\newtheorem{example}[theorem]{Example}

\theoremstyle{claim}
\newtheorem{claim}{Claim}[theorem]
\newtheorem{fact}{Fact}[theorem]
\crefname{section}{Sec.}{Secs.}
\Crefname{section}{Section}{Sections}
\crefname{equation}{Equ.}{Equs.}
\Crefname{equation}{Equation}{Equations}
\crefname{appendix}{App.}{Apps.}
\Crefname{appendix}{Appendix}{Appendices}
\crefname{subsection}{Sec.}{Secs.}
\Crefname{subsection}{Section}{Sections}
\crefname{subsubsection}{\S\!}{\S\!}
\crefname{theorem}{Thm.}{Thms.}
\Crefname{theorem}{Theorem}{Theorems}
\crefname{lemma}{Lem.}{Lems.}
\Crefname{lemma}{Lemma}{Lemmata}
\crefname{relemma}{Lem.}{Lems.}
\Crefname{relemma}{Lemma}{Lemmata}
\crefname{fact}{Fact}{facts}
\Crefname{fact}{Fact}{Facts}
\crefname{corollary}{Cor.}{Cors.}
\Crefname{corollary}{Corollary}{Corollaries}
\crefname{proposition}{Prop.}{Props.}
\crefname{claim}{Claim}{claims}
\Crefname{claim}{Claim}{Claims}
\crefname{definition}{Def.}{Defs.}
\Crefname{definition}{Definition}{Definitions}
\crefname{example}{Ex.}{Exs.}
\Crefname{example}{Example}{Examples}
\crefname{remark}{Rmk.}{Rmks.}
\Crefname{remark}{Remark}{Remarks}
\crefname{figure}{Fig.}{Figs.}
\Crefname{figure}{Figure}{Figures}
\crefname{table}{Tab.}{Tabs.}
\Crefname{table}{Table}{Tables}
\crefname{property}{Pty.}{Pties.}
\Crefname{property}{Property}{Properties}

\newcommand{\NN}{\mathbb{N}}
\newcommand{\ZZ}{\mathbb{Z}}
\newcommand{\QQ}{\mathbb{Q}}

\newcommand{\QQbar}{\overline{\mathbb{Q}}}

\newcommand{\Gen}[1]{\left\langle#1\right\rangle}
\usepackage{dsfont}
\newcommand{\GId}{\mathds{1}}

\newcommand{\KK}{K}

\newcommand{\Zcl}[2][]{\overline{#2}^{#1}}

\DeclareMathOperator{\diag}{diag}
\DeclareMathOperator{\GL}{GL}
\DeclareMathOperator{\SL}{SL}
\DeclareMathOperator{\Gal}{Gal}
\DeclareMathOperator{\Tr}{Tr}

\DeclareMathOperator{\Aut}{Aut}

\newcommand{\IdComp}[1]{{#1}^\circ}

\newcommand{\done}[1]{#1}
\newcommand{\MB}[1]{#1}
\newcommand{\MahsaShreier}[1]{\unskip}
\newcommand{\appref}[1]{\cref{#1}}

\newtheorem*{theorem*}{Theorem}
\newtheorem*{lemma*}{Lemma}
\newtheorem*{corollary*}{Corollary}

\makeatletter
\newcommand\footnoteref[1]{\protected@xdef\@thefnmark{\ref{#1}}\@footnotemark}
\makeatother

\begin{document}

\title[The Zariski Closure of Finitely Generated Matrix Groups]{On the Computation of the Zariski Closure of \\ Finitely Generated Groups of Matrices}
\author[K.~Nosan et al.]{Klara Nosan,$^1$ Amaury Pouly,$^{1,2}$ \mbox{Sylvain Schmitz,$^{1,3}$} \mbox{Mahsa
  Shirmohammadi,$^1$} \and \mbox{James Worrell$^2$}}
\address{\begin{minipage}{\textwidth}
$^1$~Universit\'e Paris Cit\'e, CNRS, IRIF, Paris, France\\
$^2$~Department of Computer Science, University of Oxford, UK\\
$^3$~Institut Universitaire de France, France\end{minipage}}
\date{}

\begin{abstract}
  \noindent
  We investigate the complexity of computing the Zariski closure of a
finitely generated group of matrices.  The Zariski closure was
previously shown to be computable by Derksen, Jeandel, and Koiran, but
the termination argument for their algorithm appears not to yield any
complexity bound.  In this paper we follow a different approach and
obtain a bound on the degree of the polynomials that define the
closure. Our bound shows that the closure can be computed in
elementary time. 
We also obtain upper bounds on the length of chains
of linear algebraic groups, where all the groups
are generated over a fixed number field.

\end{abstract}
\maketitle

\DeclareRobustCommand{\gobblefive}[5]{}
\DeclareRobustCommand{\gobblenine}[9]{}
\newcommand*{\SkipTocEntry}{\addtocontents{toc}{\gobblenine}}





\section{Introduction}

Finitely generated groups of matrices are fundamental mathematical
objects that appear in a wide variety of areas in computer science,
including algebraic complexity theory, quantum computation, dynamical
systems, graph theory, control theory, and program verification.
Unfortunately, matrix groups are challenging from the algorithmic
point of view, with many natural problems being undecidable.  For
instance, already in the case of finitely generated subgroups of the
group $\SL_4(\mathbb{Z})$ of~$4\times 4$ integer matrices with
determinant one, the membership problem and the conjugacy problem are
undecidable~\citep{Mik66,miller72}. 

In order to get a better handle on a finitely generated matrix group, it turns out to be worth over-approximating it by its Zariski closure. This closure is then a linear algebraic group and comes with a finite representation as an \emph{algebraic variety}, i.e., as the locus of zeroes of a finite collection of polynomials.
While this approximation is too coarse for solving
the membership and conjugacy problems, it can
nevertheless be extremely useful.  Indeed, we gain access to
a rich algorithmic toolbox, including Gr\"obner
bases~\citep{BeckerW93}, Lie algebras~\citep{deGraafBook}, and
decision procedures for the first-order theories of algebraically
closed fields and real closed fields~\citep{renegar92}, that allows to
manipulate the Zariski closure and solve other computational problems
on this new representation.

A first case in point comes from the field of quantum computation.  A
``measure once'' quantum finite automaton is essentially specified by
a finite collection of unitary matrices of the same dimension, and a
word is accepted if its value is strictly greater than a given
threshold.  As shown by \citet*{BJKP05}, the corresponding emptiness
problem is decidable.  This is somewhat miraculous considering that
they also show that the non-strict version of the emptiness problem is
undecidable, and considering that all versions of the language emptiness problem
are undecidable for probabilistic automata, i.e., finite automata in
which for each input letter the transition matrix is
stochastic~\cite{Paz,Blondel01}.  The miracle is underpinned by two
nontrivial facts.  First, the Euclidean closure of a group of unitary
matrices is algebraic and therefore the Euclidean closure equals the
Zariski closure.  Second, there is an algorithm due
to \citet*{DerksenJK05} to compute the Zariski closure of a finitely
generated group of matrices.  As noted in~\cite{BJKP05},
quantum automata are not the sole application of the latter result:
the question of whether a finite set of quantum gates is
universal can be reduced to computing the Zariski closure of the
associated group of unitary matrices.  Decidability then follows by
noting again that the (Euclidean) density of the group generated by
those matrices can be tested on its Zariski closure using standard
manipulations of algebraic varieties.

A second case in point comes from the field of program verification,
where program invariants are a fundamental technique that allow one to
compute an overapproximation of the set of reachable states (which is
itself not-computable in general) and thereby verify whether a program
meets a certain specification, such as loop termination.  The crux of
these techniques lies in being able to compute precise enough
over-approximations \citep{BradleyM07}.  A landmark \citeyear{Karr76}
result by \citet*{Karr76} shows that the strongest affine invariant (or
smallest linear variety, in the language of algebraic geometry) is
computable in the case of so-called affine programs; those are
programs having only non-deterministic (as opposed to conditional)
branching and all of whose assignments are given by affine
expressions
.  The problem of computing polynomial invariants rather than the
coarser affine ones has been studied extensively and was finally
answered positively in \citep*{HrushovskiOP018}, wherein some of the
authors extended the result of \cite{DerksenJK05} to compute the
Zariski closure of finitely generated semigroups.

\subsection{Contributions}  The key result in the applications we
mentioned above is the algorithm computing the Zariski closure of a
finitely generated group of matrices due to \citet*{DerksenJK05}.
Unfortunately, it is not clear how to analyse the complexity of this
algorithm, or even whether it has an elementary complexity
\done{(see the discussion in \cref{sub-DJK} and \appref{app-DJK})}.

\bigskip
In this paper, we show that the Zariski closure of a finitely
generated group of matrices is computable in elementary time.  More
precisely, we obtain a
tenfold\footnote{\label{note1}
The formulation here corrects the published version \cite{zariskipublished} by correcting the cited result {\citealt[Prop.~B.6]{Feng}} restated in \Cref{th:homomorphism_quotient}.} 

exponential bound in the dimension and
the number of the generators and the size of their entries (their
height).  We also obtain a sevenfold\footnoteref{note1}
exponential bound in the
case where the group consists of orthogonal (or more generally
unitary) transformations, as is always the case in quantum
computing.

\paragraph{Degree Bound}
Our main technical result is a quantitative structure lemma for
algebraic groups, which may be of independent interest.  It allows us
to prove the following upper bound on the degree of the polynomials
defining the Zariski closure of a finitely generated matrix group.
Here, the \emph{height} of a rational number $\frac{a}{b}$, with $a,b$
coprime integers, is $\max(|a|,|b|)$%
\MB{, and its bitsize is
  the logarithm of its height.  We denote by $\mathrm{exp}^k(x)$}  
  iterated
exponentiation in base~2, defined by $\mathrm{exp}^0(x):=x$ and
$\mathrm{exp}^{k+1}(x):=2^{\mathrm{exp}^k(x)}$.

\begin{restatable*}[Degree bound]{theorem}{thbound}\label{th:bound_finitely_generated_closure}
    Let $n\in\NN$ and let $S\subseteq\GL_n(\QQ)$ be a finite set of
    matrices whose entries have height at most
    $h \in \mathbb{N}$. Then the Zariski closure of the group
    generated by~$S$ can be represented by finitely many polynomials
    of degree at most $(\log h)^{2^{|S|^{\exp^6(\mathrm{poly}(n))}}}$
                      with coefficients in~$\QQ$, forming a basis of the vanishing
    ideal of the group generated by~$S$.
    Furthermore, if~$G$ contains only semisimple elements
    then the degree can be bounded 
    by $(\log h)^{2^{|S|^{\exp^3(\mathrm{poly}(n))}}}$.
\end{restatable*}
\noindent
\Cref{ex:dim,ex:height} in \cref{sec:discussion} show that the dependencies of
the bound on the height~$h$ and the dimension~$n$ are unavoidable.

Our approach to proving \cref{th:bound_finitely_generated_closure}
actually has more in common with an algorithm by \citeauthor*{HR} for
computing the Galois group of a linear differential
equation~\cite{HR}, than with the original algorithm
from~\citep{DerksenJK05}.  
Although \citeauthor*{HR}'s algorithm does not directly imply an
algorithm to compute the Zariski closure of a finitely generated
matrix group, our proof relies on a key insight about the structure of
linear algebraic groups from~\cite{HR}, namely the existence of a
finite-index subgroup of computable degree that lies inside a given
linear algebraic group; see \cref{sec:overview} for an overview,
and \crefrange{sec:semisimple}{sec:main} for the main proof arguments.

\paragraph{Complexity of Computing the Zariski Closure}
Now endowed with an \emph{a priori} bound~$d$ on the
degree of the polynomials that define a Zariski closure of a finitely
generated matrix group, the closure itself can be computed \done{using the
following algorithm.}
\begin{enumerate}
\item There is a trivial reduction from computing the Zariski closure
  of a matrix group generated by a finite set~$S$ to computing the
  strongest polynomial invariant of an affine
  program~\cite{Muller-OlmS04b}. Indeed, the group generated by~$S$
  is the set of reachable configurations of an affine program with a
  single location and a single loop containing a nondeterministic
  choice between applying each generator in~$S$; the strongest
  polynomial invariant of the program is then exactly the Zariski
  closure of the group generated by~$S$.
\item \Citet*{Muller-OlmS04b} show how to reduce the
  problem of computing the strongest polynomial invariant of degree at
  most~$d$ of an affine program with $m:=n^2+1$ variables to that of computing
  the strongest affine invariant in higher dimension
  , and then apply a simplification of \citeauthor*{Karr76}'s
  algorithm~\cite{Karr76}, with an overall complexity \done{of}~$O(|S|\cdot
  m^{3d})$ \done{in the uniform cost model (i.e., every machine operation
  has constant cost)}.
\end{enumerate}
\noindent
\done{The above algorithm together with the degree bound of \cref{th:bound_finitely_generated_closure} yields our main result.}

\begin{theorem}[Main result]\label{the:main}
  Let $n\in\NN$ and let $S\subseteq\GL_n(\QQ)$ be a finite set of
  matrices whose entries have height at most $h \in \mathbb{N}$.  Then
  there is an algorithm that computes the Zariski closure of the group
  generated by $S$ in time $2^{(\log h)^{2^{|S|^{\exp^6(\mathrm{poly}(n))}}}}$%
  . The closure is
  represented by a finite basis of the vanishing ideal of the group generated by~$S$.
\end{theorem}


\paragraph{Chains of Linear Algebraic Groups}
Finally, we also obtain an upper bound on the length of chains of
linear algebraic groups generated over a fixed number field;
see \cref{sec:chain}.  This bound gives a quantitative version of the
Noetherian property of subgroups of linear algebraic groups.  In the
case when all groups are generated by unitary matrices, the
bound is triply\footnoteref{note1}
 exponential in the dimension of the matrices
and degree of the number field over~$\QQ$.

\subsection{Overview}\label{sec:overview}
Let us give a high-level account of the proof of the degree bound
in~\cref{th:bound_finitely_generated_closure}.  We also introduce some
of the main definitions and notations along the way; we refer the
reader to \cite{HumphreysLAG} for more details.  Throughout this
section, we state the general definitions for a perfect field~$k$ with
algebraic closure~$K$, but our results will be framed in the case of
$k=\QQ$ the field of rational numbers and $\KK=\QQbar$ the field of
algebraic numbers.

\subsubsection{Algebraic Geometry}\label{sub:ag}
Following a standard practice (see \cref{sub:lag}), we are going to
view our finitely generated matrix group as a subset $G\subseteq\KK^m$
of the affine space over~$\KK$ for some dimension~$m$.  Our
overarching goal is to over-approximate~$G$ by
\begin{equation}\label{eq:variety}
  \mathbf V(I) := \{ (a_1,\ldots,a_m)\in \KK^m : \forall f \in I,
  f(a_1,\ldots,a_m)=0\}
\end{equation}
the set of common zeroes of a set~$I\subseteq\KK[x_1,\ldots,x_m]$ of
polynomials from the ring $\KK[x_1,\ldots,x_m]$ of polynomials in
variables $x_1,\ldots,x_m$ and coefficients in $\KK$
.  Indeed, an (affine) \emph{algebraic set} or \emph{algebraic
  variety}~$X \subseteq \KK^m$ is a set of the form $X=\mathbf V(I)$
for some $I\subseteq \KK[x_1,\ldots,x_m]$; see
\cite[Chap.~1]{HumphreysLAG}. In particular, $X$ is
\begin{itemize}
\item\emph{defined over~$k$} if there exists a set $I\subseteq
  k[x_1,\dots,x_m]$ such that $X=\mathbf V(I)$,
\item\emph{$d$-bounded} if there exists
  a set~$I\subseteq\KK[x_1,\ldots,x_m]$ of polynomials all of degree at
  most~$d\in \mathbb{N}$ such that $X=\mathbf V(I)$,
\item\emph{$d$-bounded over $k$} if there
  exists a set~$I\subseteq k[x_1,\dots,x_m]$ of polynomials all of
  degree at most~$d\in \mathbb{N}$ such that $X=\mathbf V(I)$.
\end{itemize}
\Cref{eq:variety} provides a finitely represented approximation of~$G$
since, by Hilbert's Basis Theorem, there exists a finite subset
$I'\subseteq I$ such that $\mathbf V(I)=\mathbf V(I')$.  We want to
compute the best possible such approximation, using the Zariski
topology on $K^m$ (\MB{that is, the topology whose closed sets are
  the algebraic sets}): the \emph{Zariski closure} $\Zcl G$ of a set
$G\subseteq\KK^m$ is the smallest algebraic set that contains~$G$.

\subsubsection{Linear Algebraic Groups}\label{sub:lag}
A \emph{linear group} over~$k$ is a subgroup of the \emph{general
  linear group} $\GL_n(k)$ of all invertible $n\times n$ matrices with
entries in~$k$.  We embed the latter as an algebraic set in $K^{m}$
of dimension $m:=n^2+1$ by writing
\begin{equation}\label{eq:gln}
  \GL_n(K):=\{(M,y)\in K^{n^{2}+1}: \det(M)\cdot y=1\}\;.
\end{equation}
A \emph{linear algebraic group} is a Zariski-closed subgroup of
$\GL_n(\KK)$
~\citep[Chap.~7]{HumphreysLAG}.
Given algebraic sets $X\subseteq K^m$ and $Y\subseteq K^p$, a
\emph{regular map} from $X$ to $Y$ is a function $\Phi\colon
X\rightarrow Y$ defined by~$p$ polynomials $f_1,\ldots,f_p \in
K[x_1,\ldots,x_m]$ by $\Phi(\boldsymbol{a}) :=
(f_1(\boldsymbol{a}),\ldots,f_p(\boldsymbol{a}))$ for all
$\boldsymbol{a} \in X$; regular maps are (Zariski)
continuous~\citep[Sec.~1.5]{HumphreysLAG}.  Under the
identification in~\eqref{eq:gln}, matrix multiplication is a regular
map $\GL_n(K)\times \GL_n(K) \rightarrow \GL_n(K)$, and, by Cramer's
rule, matrix inverse is also a regular map $\GL_n(K) \rightarrow
\GL_n(K)$. 

\subsubsection{Irreducible Sets and Connected Components}\label{sub:cc}
An algebraic set $X\subseteq K^m$ is said to be \emph{irreducible} if
it is non-empty and cannot be written as the union $X=X_1\cup X_2$ of
two algebraic proper subsets $X_1$ and $X_2$.  Because the Zariski
topology is Noetherian, every algebraic set~$X$ can be written as the
finite union of its \emph{irreducible components}, i.e., of its
maximal irreducible subsets.  The closure $\Zcl{\Phi(X)}$ of the image
of an irreducible set~$X$ under a regular map~$\Phi$ is again
irreducible~\citep[Sec.~1.3]{HumphreysLAG}.

As any linear algebraic group $G$ is also an algebraic set, it is the
finite union of its irreducible components, which turn out to be disjoint.
We denote by $\IdComp{G}$ the irreducible component of $G$ containing
the identity, called its \emph{identity component}.  Then $\IdComp{G}$
is a normal subgroup of $G$ and the irreducible components of $G$ are
the cosets of $\IdComp{G}$ in~$G$, called its \emph{connected
  components}~\citep[Sec.~7.3]{HumphreysLAG}.

\subsubsection{Jordan-Chevalley Decomposition}
\label{sub:la}
The starting point in the proof of
\cref{th:bound_finitely_generated_closure} is the following
decomposition of invertible matrices.  Recall that a matrix $g \in
\GL_n(k)$ is called \emph{nilpotent} if there exists some positive
integer $p$ such that~$g^p = 0$, it is \emph{unipotent} if $g - \GId$
is nilpotent (where $\GId$ denotes the identity matrix), and
\emph{semisimple} if it is diagonalisable over~$K$.  The
\emph{Jordan-Chevalley decomposition} writes a matrix~$g$ as $g = g_s
+ g_n$, where $g_s$ is semisimple, $g_n$ is nilpotent, and~$g_s$
and~$g_n$ commute.  It follows that $g=g_sg_u$, where
$g_u:=\GId+g_s^{-1}g_n$ is unipotent, and~$g_s$ and~$g_u$ commute. This
decomposition is unique.  For a linear algebraic group~$G$, we denote
by $G_s:=\{g_s:g\in G\}$ and $G_u:=\{g_u:g\in G\}$ the semisimple and
unipotent parts of~$G$, respectively. It is a standard fact that
both~$G_s$ and~$G_u$ are closed sets contained
in~$G$~\citep[Sec.~15.3]{HumphreysLAG}, but note that they might not
be groups.

\subsubsection{Main Argument (\Cref{sec:main})}
\label{sub:main}
\MB{Let $S\subseteq\GL_n(\QQ)$ be a finite set of rational matrices and
consider the closure $G=\Zcl{\Gen{S}} \leq GL_n(\QQbar)$ of
the group generated by $S$.}
Let $U:=\Zcl{\Gen{G_u}}$ be the closure of
  the group generated by~$G_u$ as per~\cref{sub:la}; note
that~$U$ 
might contain non-unipotent elements
.
Our main technical result proven in \cref{sec:general} and
summarised next in \cref{sub:general} identifies a subgroup~$H$ of~$G$
with good properties.
\MahsaShreier{Do we require finiteness of $S$ for Lemma 7? It seems to me that we only require this assumption when we want to compute the degree bound and need to compute a finite set of generators for $H'$? Does the obtained structural property holds for all (and not necessarily finitely-generated) algebraic groups?  }

\begin{restatable*}[Quantitative Structure Lemma]{lemma}{quantitative}\label{lem:hrushovski_alggroup}
  Let $n\in \NN$, $S\subseteq \GL_n(\QQ)$ be a \MahsaShreier{finite?} set of
    matrices, \done{and $G=\Zcl{\Gen S}$}. Write $U:=\Zcl{\Gen{G_u}}$ for the closure of the
  subgroup generated by the unipotent elements of~$G$.  Then there
  exists a \done{closed} normal subgroup $H$ of $G$ such that $U \trianglelefteq
  H \trianglelefteq G$ and
  \begin{enumerate}[(i)]
  \item\label{quant1} $H / U$ is commutative,
  \item\label{quant2} $H$ is a normal subgroup of $G$ of index at
    most $J'(n):=(2^{3+2^{4(n^2+D)^D + 1}})!$,
    where $D:=(n^3+1)^{2^{3n^2}}$.
  \end{enumerate}
\end{restatable*}

We exploit these two properties in \cref{sec:main} to obtain a degree
bound for~$G$ as follows.
\begin{enumerate}
\item We find a finite set of generators~$S'$ for~$H=\Zcl{\Gen{S'}}$
  using the finite-index property.  This relies on Schreier's Lemma;
  see
  \cref{lem:finite_index_subgroup_generators} in~\cref{sec:main}. 
\item We deduce a degree bound for $H$ using the commutativity
  of~$H/U$.  Indeed, we have that $$H=\Zcl{\big(\prod_{g\in
      S'}\Zcl{\Gen{g}}\big)\cdot U}\;.$$ In other words, we only need
  to compute the closure of individual elements $\Zcl{\Gen{g}}$ for
  $g\in S'$.  The latter can be done using \citeauthor*{Mas88}'s bound on the
  group of multiplicative relations among the eigenvalues of a
  rational matrix; see \cref{th:masser,claim:main-bound-h}
  in~\cref{app:main}.
\item Finally, by writing~$G$ as the union of~$J'(n)$ cosets of~$H$, we
  obtain a degree bound on~$G$ from the degree bound on~$H$.
\end{enumerate}

\subsubsection{Quantitative Structure Lemma (\Cref{sec:unipotent,sec:general})}
\label{sub:general}
Here we borrow several insights from \citeauthor*{HR}'s approach on
proto-Galois groups~\citep{HR} and \citeauthor*{Feng}'s complexity
analysis of the algorithm~\citep{Feng}.  First of all, and drawing
from the works of \citeauthor*{HR} and \citeauthor*{Feng}, we show in
\cref{sec:unipotent} that $U=\Zcl{\Gen{G_u}}$ is normal
in~$G$ and that one can write polynomials for~$U$ whose degree
depends only on the dimension~$n$.

\MB{ We use the fact that the quotient of the linear algebraic
  group~$G$ by its normal subgroup~$U$ is a linear algebraic group of
  higher dimension, which can be obtained through the \emph{quotient
  map}~$\phi_U\colon
  G\to\GL_p(\QQbar)$~\citep[Sec.~11.5]{HumphreysLAG}.  This is a
  regular map that is also a group homomorphism (thus
  $G/\ker\phi_U\cong \phi_U(G)$ by the First Isomorphism Theorem) such
  that $U=\ker\phi_U$.}  Furthermore, since~$U$ contains~$G_u$, we
know that the quotient~$G/U$ consists solely of semisimple elements
\MB{when viewed 
  as a subgroup
  of $\GL_p(\QQbar)$}.  Then we can apply the following lemma,
  summarised next in \cref{sub:semisimple}, that yields another linear
  algebraic group~$P$, called the \emph{pistil} of~$G/U$.

\begin{restatable*}[Pistil Lemma]{lemma}{pistil}\label{lem:hrushovski_semisimple}
Let $n\in \NN$.
Let $G \leqslant \GL_n(\QQbar)$  be the closure of a group generated by rational matrices and assume that $G$ only consists of semisimple elements.
Then
    there exists \done{a closed group} $P\leqslant\GL_n(\overline{\QQ})$ such that 
    \begin{enumerate}[(i)]
    \item\label{pistil1} $P$ is commutative,       
     \item\label{pistil2} $G \cap P$ is a normal subgroup of $G$ of index at most 
        $J(n):=(2^{2n^2+3})!$.
    \end{enumerate}
\end{restatable*}


The idea is then to ``pull-back'' the pistil~$P$ from~$G/U$
\done{through the quotient map~$\phi_U$} to the
original group~$G$.  
  We use the following quantitative \done{statement} 
with a bound on the degree of the polynomials defining the quotient
map and a bound on the dimension.


\begin{theorem}[{\citealt[Prop.~B.6]{Feng}}]\label{th:homomorphism_quotient}
  Let $d,n\in\NN$ and let $G\subseteq\GL_n(K)$ be a closed group.
  Let $H$ be a normal subgroup of $G$ that is $d$-bounded over $k$.
  Then there exists $p\leqslant 2^{2(n^2+d)^d}$ and a homomorphism
  $\phi_H\colon G\to\GL_p(K)$ defined by polynomials over~$k$ of
  degree bounded by $d(n^2+d)^d2^{2(n^2+d)^d}$ such that $H=\ker\phi_H$.
\end{theorem}


\begin{figure}[tbp]
    \centering
    \scalebox{1}{\begin{tikzpicture}[auto,scale=.8,
            app/.style={->},
            consequence/.style={->,green!65!black,decorate,decoration={coil,aspect=0,post length=0.5ex}},
            prop/.style={->,blue!80!blue,dashed}
            ]
        \draw node (gln) {$\GL_n(\QQbar)$};
        \draw node[right=2.6cm of gln] (glp) {$\GL_p(\QQbar)$};
        \draw node[right=3cm of glp] (glm) {$\GL_q(\QQbar)$};
        \draw node[below=.6cm of gln] (g) {$G$};
        \draw (g -| glp) node (g/u) {$G/U$};
        \draw (g -| glm) node (fin) {$\phi_{P}(G/U)$};
        \draw node[below=1cm of g] (h) {$H$};
        \draw (h -| g/u) node (h') {$P$};
        \draw[app] (g) -- (g/u) node[midway,above] {$\phi_U$};
        \draw[app] (g/u) -- (fin) node[midway,above] {$\phi_{P}$};
        \draw[app] (h') -- (h) node[midway,above] {$\phi_U^{-1}$};
        \draw[consequence] (g) -- (h) node[midway,left] {\cref{lem:hrushovski_alggroup}};
        \draw[consequence] (g/u) -- (h') node[midway,right] {\cref{lem:hrushovski_semisimple}};
        \draw[prop] (fin) edge[bend left=15] node[near start,inner sep=2pt,rotate=20,below]{\textcolor{blue}{finite index property}} (h');
        
        
        \draw[red!80!black, align=center] ($(g/u)!0.5!(glp)$) node {semisimple\\elements}
                            ($(fin)!0.5!(glm)$) node {finite};
    \end{tikzpicture}}
    \caption{\label{fig:high_level}Assembling the arguments of
      \cref{lem:hrushovski_alggroup,lem:hrushovski_semisimple}: we
      take two consecutive quotients to obtain a finite group.  We
      then pull-back to obtain a group~$H$ with the desired
      properties.  Note that the dimension increases after each
      quotient and the lemmata relate~$p$ and~$q$ to~$n$.
    }
\end{figure}

Using \cref{th:homomorphism_quotient} as described
in~\cref{fig:high_level}, we define $H\leq \GL_n(\QQbar)$ to be the
preimage of the pistil~$P$ under~$\phi_U$.  Thanks to the properties
of~$P$ and~$\phi_U$, $H$ has the properties desired for the proof of
\cref{th:bound_finitely_generated_closure}, namely~$H$ is a normal
subgroup of~$G$ of finite index depending only on~$n$, and~$H/U$ is
commutative.

\subsubsection{Pistil Lemma (\Cref{sec:semisimple})}
\label{sub:semisimple}
Let $G=\Zcl{\Gen{S}}\leqslant\GL_p(\QQbar)$ be \done{the closure of a group}
generated by a \MahsaShreier{finite?} set~$S$ of rational matrices. Assume that $G$
consists exclusively of semisimple elements.
\MahsaShreier{Mahsa: Do we require $S$ to be finite in above? It seems to me that we should not assume this, specially as $G/U$ might not be finitely generated? this is probably just a typo as finiteness is not assumed in statement of Lemma 4. }

Since the identity component $\IdComp{G}$ of~$G$ is connected and all elements are semisimple, it is a torus.
Recall that an algebraic group $T \leqslant \GL_n(K)$ is said to be an \emph{algebraic torus} (henceforth simply a \emph{torus}) if it is isomorphic to a product of $n$ copies of the multiplicative group $K^\times$ of the base field $K$. For a connected algebraic group  $T$ the following conditions are equivalent: (i) $T$ is a torus, (ii) $T$ consists only of semisimple elements, (iii) $T$, considered as a matrix group, can be made diagonal~\citep[Thm~3.1 on p.~3]{borel66}. Furthermore, it holds that the \emph{maximal tori} in $\GL_n(\QQbar)$ are exactly those groups that are conjugate to the group of diagonal matrices in $\GL_n(\QQbar)$. 
As
stated above in~\cref{sub:cc}, $\IdComp{G}$ is a normal subgroup
of~$G$ and $G/\IdComp{G}$ is finite.  However we do not know how to
directly bound the degree of $\IdComp{G}$ nor the cardinality
of~$G/\IdComp{G}$ in terms of the set~$S$ of generators and the
dimension~$p$.  Instead, again following~\cite{HR}, we identify
another group~$P\leqslant\GL_p(\QQbar)$, dubbed the \emph{pistil}
of~$G$, and defined as the intersection of all the maximal tori that
contain~$\IdComp{G}$.  This construction is 
explained in \cref{sec:semisimple}.  Since maximal tori
are 1-bounded, commutative\done{, and closed}, $P$ is also 1-bounded,
commutative\done{, and closed}. Observe that the quotient $G/(G\cap P)$ is finite because
$G\cap P$ includes~$\IdComp{G}$. We can
bound the index $[G:G\cap P]$ of this quotient in terms
of~$p$, using the fact that the quotient map~$\phi_{P}$ 
defined by \cref{th:homomorphism_quotient} embeds into a
\emph{finite} subgroup of~$\GL_q(\QQbar)$ for some~$q$ depending only
on~$p$; see \cref{fig:high_level}.

\subsubsection{About Field Extensions}\label{sub:prelim}
The reader might have noticed that, in order to apply
\cref{th:homomorphism_quotient}, we need to show that $H$ is
$d$-bounded \emph{over~$\QQ$}.  We resort routinely to the following
folklore lemma, where $L/k$ is any field extension and $\Aut(L/k)$ the
group of automorphisms of~$L$ that fix~$k$ (see \appref{app:prelim}
for a proof).
\begin{restatable}{relemma}{boundkgroup}\label{lem:bound_k_group_is_k_bounded}
  Fix $m\in \NN$, and let $X\subseteq L^m$ be a $d$-bounded algebraic
  set. Then $X$ is $d$-bounded over~$k$ if any of the following holds:
  \begin{enumerate}[(i)]
  \item\label{itm:k_bounded} $X$ is stable under~$\Aut(L/k)$ for some subfield $k$ of $L$, or
  \item\label{itm:dense} $X\cap k^m$ is dense in $X$.
  \end{enumerate}
\end{restatable}

\section{Pistil Lemma}
\label{sec:semisimple}


Recall that a matrix $g\in \GL_n(\overline{\QQ})$ is \emph{semisimple}
if it is diagonalisable over~$\overline{\QQ}$.  In this section, we
study structural properties of algebraic groups~\done{$G=\Zcl{\Gen S}$
defined as closures of groups of semisimple rational matrices.} 
(We note in passing that such a group $G$ need not be ``semisimple as a
subgroup'' of $\GL_n(\overline{\QQ})$, that is, conjugate to a group
of diagonal matrices.)  We show 
that such groups are virtually abelian: they contain an abelian
subgroup of finite index.  We define \done{for this} a
group~$P\leqslant\GL_n(\QQbar)$, which we call the \emph{pistil}
of~$G$, and show that $P\cap G$ is a normal abelian subgroup of~$G$
with finite index.


\pistil

\begin{proof}
The group $P$ will be defined to encompass the identity component $\IdComp{G}$ of $G$. 
%
By definition, $\IdComp{G}$ is connected, and by assumption, $\IdComp{G}$ is comprised exclusively of semisimple elements, hence it is a torus~\cite[Thm~3.1 on p.~3]{borel66}.
%
\begin{figure}
    \centering
    \scalebox{1}{\begin{tikzpicture}
        \begin{scope}[scale=0.33]
            \draw[green!65!black,rotate=-7] (0,-4.65) ellipse[x radius=1.9, y radius=5.6] (0,-5.5) node {$G$};
            \draw (0,-0.5) circle[radius=1] node{$\IdComp{G}$};
            \begin{scope}[shift={(0,0.3)}]
                \draw[red!70!black] {[rotate=30] (2.5,0) ellipse[x radius=5, y radius=3]}
                    {[rotate=-30] (2.5,0) ellipse[x radius=5, y radius=3]}
                    {[rotate=70] (2.5,0) ellipse[x radius=5, y radius=3]}
                    {[rotate=100] (2.5,0) ellipse[x radius=5, y radius=3]}
                    {[rotate=150] (2.5,0) ellipse[x radius=5, y radius=3]}
                    {[rotate=210] (2.5,0) ellipse[x radius=5, y radius=3]};
                \clip [rotate=-30] (2.5,0) ellipse[x radius=5, y radius=3];
                \clip [rotate=30] (2.5,0) ellipse[x radius=5, y radius=3];
                \clip [rotate=70] (2.5,0) ellipse[x radius=5, y radius=3];
                \clip [rotate=100] (2.5,0) ellipse[x radius=5, y radius=3];
                \clip [rotate=150] (2.5,0) ellipse[x radius=5, y radius=3];
                \clip [rotate=210] (2.5,0) ellipse[x radius=5, y radius=3];
                \path[fill=orange!90!yellow,fill opacity=0.25] (-3,-3) rectangle (3,3);
                \node[orange!80!red] at (0,1.5) {$P$};
            \end{scope}
        \end{scope}
        \begin{scope}[shift={(4,1)}]
            \draw (0,0) node[anchor=west] (top) {\cref{lem:hrushovski_semisimple}};
            \foreach \i/\txt in {1/$P\leqslant\GL_n(\overline{\QQ})$,
                                 2/$\IdComp{G} \vartriangleleft (G\cap P)\vartriangleleft G$,
                                 3/$[G:G\cap P]\leqslant J(n)$,
                                 4/{$P$ is commutative},
                                 5/{$P$ is 1-bounded}} {
                \draw[anchor=west] (0.65,-\i*3.5ex) node (n-\i) {\txt};
                \draw[->] ($(top.south west)+(2ex,0)$) |- (n-\i);
            }
        \end{scope}
    \end{tikzpicture}}
    \caption{\label{fig:illustration_hrushovski_semisimple} Graphical
      representation of \cref{lem:hrushovski_semisimple}: given a
       group $G$ that is the closure of a group generated by rational matrices and is comprised of semisimple elements, we define the pistil~$P$ as the intersection of the
      maximal tori containing~$\IdComp{G}$; note that~$P$ is not
      necessarily a subgroup of~$G$.}
\end{figure}  
%
%
Now $\IdComp{G}$, being a torus, is contained in some maximal torus in
$\GL_n(\QQbar)$.  We define the \emph{pistil} of~$G$ to be the
intersection of all maximal tori in $\GL_n(\QQbar)$ that contain
$\IdComp{G}$;
see \cref{fig:illustration_hrushovski_semisimple}, where the maximal
tori containing~$\IdComp{G}$ are depicted as red ellipses. 
 We now show that the pistil~$P$ has all the announced properties.
\medskip


%
\done{First note that, as an intersection of closed groups, $P$ is
closed.}

Regarding~\ref{pistil1}, observe that the commutativity of $P$ follows from
the fact that it is obtained by intersecting commutative groups.


Regarding~\ref{pistil2}, 
to prove that $G \cap P$ is a normal subgroup of~$G$, it suffices to 
\color{black}
observe that the collection of all maximal tori
containing~$\IdComp{G}$ is stable under conjugation by all
elements~$g\in G$, and hence so is~$P$.
More precisely, fix $g \in G$
  and observe that for every maximal torus~$T$ containing
  $\IdComp{G}$, $gTg^{-1}$ is also a maximal torus, which moreover
  contains~$\IdComp{G}$ because $g\IdComp{G}g^{-1}=\IdComp{G}$ since $\IdComp{G}$ is a normal subgroup of~$G$.  We thus
  see that conjugation by~$g$ induces a bijective self-map on the
  collection of maximal tori containing~$\IdComp{G}$.  It follows that
  $gPg^{-1}=P$ for all~$g \in G$, i.e., $G$ normalises~$P$ and hence
  $G\cap P$ is a normal subgroup of~$G$.
\medskip

The main difficulty of the proof lies
in computing an upper bound
on the index~$[G:G\cap P]$. For
this, we will argue that the quotient $G/ (G \cap P)$ can be
embedded as a finite subgroup of rational matrices in~$\GL_p(\QQ)$
for some $p$. Recall from~\cite{KP2002} that such a finite subgroup of rational matrices is
conjugate to a finite subgroup of integer matrices, and furthermore
that the order of any such finite subgroup of integer matrices divides
$(2p)!$~\cite{Newman72,Guralnick2005OrdersOF,BumpusHKST20}
.\footnote{%
\done{\Citet*{Friedland97} showed that the better bound $2^nn!$
holds for all $n\geq 72$ but this relies on an incomplete manuscript
by \citet*{Weisfeiler85}.}} 

\begin{fact}\label{theo-order-finite-G}
  For all $p\in \NN$, every finite subgroup of $\GL_{p}(\QQ)$ has
  order at most $(2p)!$.
\end{fact}

\noindent By
\cref{theo-order-finite-G}, 
the
size of the quotient $G/ (G \cap P) \in \GL_p(\QQ)$ depends on the
dimension~$p$, which in turn by~\cref{th:homomorphism_quotient}
depends on~$n$ and the maximal degree of the defining polynomials of $G \cap
P$. The latter is not directly accessible to us; however, we
first show the following claim 
and will later bound~$p$ through the pistil and
its normaliser.

\begin{claim}
  \label{claim:pistil-bounded}
  $P$ is $1$-bounded over $\QQ$.
\end{claim}

\begin{claimproof}[Proof of \cref{claim:pistil-bounded}]
Recall that all maximal tori
of~$\GL_n(\QQbar)$ are conjugate to the group of all diagonal matrices
in $\GL_n(\QQbar)$ and hence $1$-bounded. As~$P$ is defined to be the
intersection of maximal tori, it follows that $P$ is $1$-bounded too.
  We further show that $P$ is stable under $\Aut(\QQbar/\QQ)$.
  Indeed,
  since~$G$ is defined over~$\QQ$, it is stable under the action
  $\Aut(\overline{\QQ}/\QQ)$ on~$\QQ^{n^2+1}$. Now observe that all
  automorphisms $\sigma\in\Aut(\QQbar/\QQ)$ preserve connectedness,
  and thus permute the connected components of~$G$. Note also that all
  automorphisms fix the identity matrix, and therefore must fix the
  connected component containing it, that
  is~$\IdComp{G}$.
  Finally, the automorphisms
  in~$\Aut(\QQbar/\QQ)$ also preserve maximality, and hence permute
  all maximal tori of~$\QQbar$. It follows that~$P$, being the
  intersection of all maximal tori of~$\GL_n(\QQ)$, is stable under
  $\Aut(\QQbar/\QQ)$, and thus we can apply
  \cref{lem:bound_k_group_is_k_bounded}\ref{itm:k_bounded}
  to conclude the proof.
\end{claimproof}

It remains to show
that the index of the normal subgroup $G \cap P$ in~$G$ is bounded
by~$J(n)$. 
Let $N:=\{g\in\GL_n(\overline{\QQ}):gP=Pg\}$ be the \emph{normaliser}
  of~$P$ in~$\GL_n(\QQbar)$,
that is, the largest subgroup of~$\GL_n(\QQbar)$ that normalises~$P$.
In particular, $G$ is a subgroup of~$N$.





  
  Since $P$ is defined over~$\QQ$, by
  quantifier elimination, $N$ is also defined over~$\QQ$.
  Since $P$ is both $1$-bounded over~$\QQ$
  (\cref{claim:pistil-bounded}) and a normal subgroup of $N$,
  by \cref{th:homomorphism_quotient}, there is a homomorphism
  $\phi_P\colon N\to\GL_p(\QQbar)$, for some $p\leq
  2^{2(n^2 + 1)}$, such that $\phi_P$ has kernel~$P$ and is defined
  over~$\QQ$.
  
  As $G \leqslant N$, 
   we can look at the
image of~$G$ under~$\phi_P$.  
Since
$\IdComp{G}$ has finite index in~$G$ and $\IdComp{G}\subseteq P$, 
$G \cap P$ also has finite index in~$G$, and therefore $\phi_P(G)$ is
finite.
  
    
  Furthermore, since $G=\Zcl{\Gen S}$ \done{for $S\subseteq{\GL_n(\QQ)}$} and $\phi_P$ is
  defined over~$\QQ$, the rational points of~$\phi_P(G)$
  are dense in~$\phi_P(G)$, that is,
  $\Zcl{\phi_P(G)\cap\QQ^{{p^2+1}}}=\phi_P(G)$. Since~$\phi_P(G)$ is finite,
  $\Zcl{\phi_P(G)\cap\QQ^{{p^2+1}}}=\phi_P(G)\cap\QQ^{{p^2+1}}=\phi_P(G)$ is a finite
  subgroup of~$\GL_p(\QQ)$.
The bound $J(n)$ on the index of $G\cap P$ now follows
  from \cref{theo-order-finite-G}.
\end{proof}


\section{Unipotent-Generated Groups}
\label{sec:unipotent}

Recall that a matrix $g\in \GL_n(\QQbar)$ is unipotent if $(g-\GId)^n=0$.
In this section, we investigate for a given algebraic group $G$ the
\done{closure}~$U$ \done{of the group} generated by the unipotent elements of $G$.  We remark that
in general $U$ need not consist exclusively of unipotent elements
(that is, $U$ differs in general from the so-called unipotent radical
of $G$).
%
%
We use instead the fact, noticed by \citet{HR} and
analysed more precisely by \citet{Feng}, that the \done{closure of the}
group generated by
the unipotent elements shares underlying similarities with the
unipotent radical.

Before proceeding with the proof
we recall the
definition of matrix exponential and matrix
logarithm~\citep[Sec.~15.1]{HumphreysLAG}. Given a matrix
$g\in \GL_n(\overline{\QQ})$, its exponential is
defined by
$\exp(g):=\sum_{n=0}^{\infty} \frac{g^n}{n!}$.

A matrix $g'$ is said to be a \emph{matrix logarithm} of~$g$ if $\exp(g) =
 g'$. Since the exponential function is not one-to-one in~$\QQbar$,
 some matrices may have more than one logarithm.  However, for a
 unipotent matrix~$g$ we can define the logarithm uniquely as
$\log(g):=\sum_{n=1}^{\infty} (-1)^{n-1}\frac{(g-\GId)^n}{n}.$

We also use the fact that, given a polynomial ideal, we have the
following degree bound due to \citet{Dube90} on the generators of the elimination ideal
obtained by eliminating some of the variables.

\begin{proposition}[{\cite[Cor.~8.3]{Dube90}}]\label{prop:elm=ideal}
	Let $I \subseteq K[x_1,\ldots,x_n]$ be an ideal generated by
	polynomials with degree at most~$d$.  \done{Then} for all $i\in \{1,\ldots,n\}$ the elimination ideal $I\cap K[x_1,\ldots,x_i]$ is generated by polynomials with degree at most~$(d+1)^{2^n}$. 
\end{proposition}

  \begin{lemma}\label{lem:subgroup_generated_by_unipotents}
    Let $n\in \NN$ and $G\leqslant\GL_n(\overline{\QQ})$ be defined over $\QQ$. Let $U=\Zcl{\Gen{G_u}}$ be the closure of the subgroup
    generated by the unipotent elements of $G$.  Then $U$ is a normal subgroup of $G$ which is $(n^3+1)^{2^{3n^2}}$-bounded over $\QQ$.
\end{lemma}

\begin{proof}
 Consider an arbitrary $g \in G$. Observe that the set $G_u$ is
  closed under conjugation by~$g$. Therefore $\Gen{G_u}$ is also
  closed under conjugation by~$g$.
  By Zariski continuity of conjugation,
  the \done{closure}~$U=\Zcl{\Gen{G_u}}$ \done{of the group} generated by the unipotent elements of~$G$ is  also closed under conjugation
  by~$g$
  , i.e.,
  $gUg^{-1}=U$.  Since $g\in G$ was arbitrary, $U$ is a normal
  subgroup of~$G$.
 
%

The degree bound in the following claim comes from~\cite[Lemma B.8]{Feng}; note that it
only depends on the dimension. 

\begin{restatable}{claim}{Uconnectedbounded}
\label{claim:U-connected-bounded}
	$U$ is  irreducible and $(n^3+1)^{2^{3n^2}}$-bounded over~$\QQ$.
\end{restatable}

\begin{claimproof}[Proof of \cref{claim:U-connected-bounded}]
  By~\cite[Lem.~4.3.9]{deGraafBook}, for any $h \in G_u$, there is a
  unique nilpotent matrix~$\log(h)$ such that $h=\exp(\log(h))$.  We
  furthermore have that~$\Zcl{\Gen{h}}$ is the \done{closure of the} image of the map
  $\Phi_h\colon\overline{\QQ}\rightarrow \GL_n(\overline{\QQ})$ given
  by $\Phi_h(z) := \exp(z \log(h))$.  \done{Indeed,} $\Phi_h$ is a
  polynomial in $z$ of degree at most~$n$ by the nilpotency
  of~$\log(h)$, hence Zariski-continuous\done{, hence
  $\Zcl{\Gen{h}}=\Zcl{\Phi_h(\ZZ)}=\Zcl{\Phi_h(\Zcl\ZZ)}=\Zcl{\Phi_h(\QQbar)}$}
  is irreducible.
    
  We recall a standard fact that given two irreducible varieties $J$
        and $J'$, the cartesian product $J \times J'$ is again
        irreducible \cite[see, e.g.][Thm.~1.1.9]{deGraafBook}. Since
        matrix multiplication is a regular map, $\Zcl{J\cdot J'}$ is
        then irreducible as well. As recalled later in \cref{sub:chains}, any strictly increasing
  chain of irreducible algebraic subsets of $\GL_n(\overline{\QQ})$
  has length at most $\dim(\GL_n(\overline{\QQ}))=n^2$.  
  Thus there
  exists~$\ell \leq n^2$ and $h_1,\ldots,h_\ell \in G_u$
  with \[U=\Zcl{\Zcl{\Gen{h_1}}\cdots \Zcl{\Gen{h_\ell}}}=\Zcl{\{ \Phi_{h_1}(z_1) \cdots \Phi_{h_{\ell}}(z_\ell)
  : z_1,\ldots,z_\ell \in \overline{\QQ} \}} \; .\] In other words, $U$ is the
  Zariski closure of the image of $\overline{\QQ}^\ell$ under a
  polynomial map of degree at most~$n^3$ with~$n^2 + 1$ variables
  defining~$U$ and $\ell$ variables used in the
  polynomials~$\Phi_{h_i}(z_i)$. Furthermore, by above it follows that $U$ is irreducible.
  

  We use \cref{prop:elm=ideal} to eliminate all~$z_i$.  
  Therefore $U$ is bounded by
  $(n^3+1)^{2^{(2n^2+1)}} \leq (n^3+1)^{2^{3n^2}}$.  
  Furthermore, since both $G$ and the collection of unipotent matrices
  in $\GL_n(\overline{\QQ})$ are stable under
  $\Aut(\overline{\QQ}/\QQ)$, $U$ is also stable under
  $\Aut(\overline{\QQ}/\QQ)$.
  By \cref{lem:bound_k_group_is_k_bounded}\ref{itm:k_bounded}, it
  follows that $U$ is bounded
  by $(n^3+1)^{2^{3n^2}}$ over $\QQ$.
\end{claimproof}
\noindent
This concludes the proof of \cref{lem:subgroup_generated_by_unipotents}.
\end{proof}

\section{Quantitative Structure Lemma}
\label{sec:general}
We now move to the general structural analysis described
in \cref{sub:general}, where we analyse an arbitrary algebraic group
to identify a normal subgroup~$H$ with the following properties.

\quantitative

\begin{proof}    

	Since~$G$ is an algebraic set, it is
  $d$-bounded for some $d\in \NN$.  We further have that the set of rational points of~$G$
  is dense in~$G$,
  hence by \cref{lem:bound_k_group_is_k_bounded}\ref{itm:dense},
  $G$ is defined over~$\QQ$.

    By \cref{lem:subgroup_generated_by_unipotents}, the
  closure~$U$ of the subgroup generated by the unipotent elements
  of~$G$ is a normal subgroup of~$G$ bounded by~$D:=(n^3 +
  1)^{2^{3n^2}}$ over~$\QQ$.
  We can thus use \cref{th:homomorphism_quotient} to
  obtain a homomorphism $\phi_U\colon G \to \GL_p(\QQbar)$ defined
  over~$\QQ$, with kernel~$U$,  where~$p \leq 2^{2(n^2+D)^{D}}$.
   The image 
  $\phi_U(G)$ of $G$ under~$\phi_U$ is isomorphic to~$G / U$
  and will not have any (non-trivial) unipotent
  elements. Indeed, if $\phi_U(g)$ were unipotent, then $g$
  would also be, thus $g\in U$, hence $\phi_U(g)=\GId$ would be the identity.
  
  We apply \cref{lem:hrushovski_semisimple} to~$\phi_U(G)$
  and examine the pistil~$P$ of~$\phi_U(G)$. Recall that
  $P$ is \done{closed} and commutative, and $\phi_U(G)\cap P$ is a normal subgroup of
  $\phi_U(G)$ of index at most~$J(p)$.  We 
  define the
  group~$H$ as the preimage of the pistil~$P$ by~$\phi_U$. That is, we
  set $H:=\phi_U^{-1}(P) \subseteq G$; see \cref{fig:general-case} for
  a depiction.  
  \done{Since~$P$ is closed, its pre-image~$H$ through the regular map
  $\phi_U$ is also closed.}  Since~$U=\ker(\phi_U)$, we have~$U\subseteq
  H$ and therefore 
    $U$ is a normal subgroup of $H$, 
  ensuring that the quotient $H / U$ is well-defined.
  
  \medskip
  Regarding~\ref{quant1}, we observe that~$\phi_U$ induces an
  injective homomorphism between~$H$ and~$P$, which is commutative,
  and therefore
    $H / U$ is also commutative.
  
  Regarding~\ref{quant2}, consider the homomorphism
  \[ \psi\colon G \to \phi_U(G) / (\phi_U(G) \cap P) \]
  obtained by composing $\phi_U$ with the natural quotient map. Since
  \[ \phi_U^{-1} (\phi_U(G) \cap P) = \phi_U^{-1}(\phi_U(G)) \cap \phi_U^{-1}(P) = G \cap H = H, \]
  the map~$\psi$ has kernel~$H$, which implies by the First
  Isomorphism Theorem that~$H$ is a normal subgroup of~$G$ and
  \begin{align*}
     	[G: H]&=[G:\ker\psi]=|\psi(G)|=[\phi_U(G):\phi_U(G)\cap P] \\[-1ex]
     		  &\leq J(p) = (2^{3+2^{4(n^2+D)^D + 1}})!\;\:.\qedhere
  \end{align*}

%
%
%
%

\end{proof}

\begin{figure}[tbp]
  \centering
    \resizebox{\figwidth}{!}{\begin{tikzpicture}[scale=1,
            pistil style/.style={thick,blue!80!black},
            ]
        \begin{scope}
            \draw[violet!90!black,thick] (0,0) to[out=90,in=180] (2,2) to[out=0,in=180] (4,1.7)
                node[above=1ex] {$G=G_u\cdot G_s$}
                to[out=0,in=90] (5,1)
                to[out=-90,in=45] (4,0) to[out=-135,in=0] (2,-2)
                to[out=180,in=-90] cycle;
            \draw[green!65!black,thick] (2,1.5) to[out=0,in=90] (4,1)
                to[out=-90,in=20] (2,-0.5) node[below] {$H:=\textcolor{black}{\phi_U^{-1}}{\color{black}(}{\color{blue!80!black}P}{\color{black})}$}
                to[out=-160,in=-90] (0.5,0.5) to[out=90,in=180] cycle;
            \draw[thick,red!90!black] (1,0.5) to[out=90,in=180] (1.5,1)
                to[out=0,in=150] node[pos=0.5,right,yshift=1ex] {$U:=\Zcl{\Gen{G_u}}$} (2.5,0.6)
                to[out=-30,in=30] (2,0) to[out=-150,in=-90] cycle;
            \coordinate (arrow end) at (4,-0.5);
        \end{scope}
        \begin{scope}[shift={(8,0)}]
            \def\phiG{(0,0) to[out=70,in=-180] (2,2) node[above]
  {$\phi_U(G) \cong G/U$}
                to[out=0,in=90] (2.5,1) to[out=-90,in=45] (1,-1) to[out=-135,in=-110] cycle};
            \def\pistil{(0.5,0.5) to[out=90,in=180] (3,1.5) to[out=0,in=90] (6,0.5)
                to[out=-90,in=0] (3,-0.5) to[out=180,in=-90] cycle};
            \draw[red!90!black,thick] \phiG;
            \draw[pistil style] \pistil (2.5,0.5) node[anchor=west]{$P=\textcolor{black}{\text{pistil of }\phi_U(G)}$};
            \coordinate (arrow start) at (0.5,-0.5);
            \begin{scope}
                \clip \phiG;
                \path[fill=green,opacity=0.2] \pistil;
            \end{scope}
            \draw[red,fill=red] (1.3,0.2) circle[radius=0.3] node[above=1.5ex,xshift=1ex] {$\IdComp{\phi_U(G)}$};
            \draw[green!65!black,semithick] (2,0.3) -- (3,-1) node[anchor=north] {$\phi_U(G)\cap P\trianglelefteq \phi_U(G)$};
        \end{scope}
        \draw[->,green!65!black,thick] (arrow start) to[out=-130,in=-50] node[pos=0.5,below] {$\phi_U^{-1}$} (arrow end);
        \draw[->,thick,black] (6,1) -- ++(1.5,0) node[pos=0.5,above=1ex] {$\phi_U\colon G\to\GL_p(\QQbar)$};
    \end{tikzpicture}}
    \caption{\label{fig:general-case} Graphical representation
  of \cref{lem:hrushovski_alggroup}: given a group
  \done{$G=\Zcl{\Gen S}$ for $S\subseteq \GL_n(\QQ)$ a set of rational
  matrices}, we identify a
  normal subgroup $H$ with the desired properties. We construct a homomorphism $\phi_U \colon G \to \GL_p(\QQbar)$ with kernel $U$, and define $H$ to be the preimage of the pistil $P$ of $\phi_U(G) \cong G/U$ under $\phi_U$.
    }
\end{figure}

\section{Degree Bounds}
\label{sec:main}


As sketched in \cref{sub:main}, we can now put everything together to
obtain a degree bound on a linear algebraic group \done{$G=\Zcl{\Gen
S}$ for $S\subseteq \GL_n(\QQ)$ a finite set of rational matrices}.  We further obtain a bound for groups consisting
exclusively of semisimple elements, which applies to groups generated
by unitary transformations as appear, e.g., in quantum automata.  

The obtained bound crucially depends on the height of the generators of the algebraic group that we are analysing.
This dependency, in turn, comes from the analysis of the normal subgroup~$H$ of $G$, as defined in \cref{lem:hrushovski_alggroup}. In particular, we need the following
lemma, which
establishes that if a group is
generated by a set~$S$ and $H$ is a finite-index subgroup, then we can
describe a set of generators of~$H$ from~$S$ and the index of the
subgroup.  Below, for $S\subseteq\GL_n(\QQ)$, we write~$S^{\leq k}$
for the set of products of at most~$k$ elements from~$S$.
\MahsaShreier{Mahsa: I think to use  Schreier's lemma we need $S$ to be a finite set. Could you please check this? }     
\begin{lemma}\label{lem:finite_index_subgroup_generators}
  Let $n\in \NN$ and $G=\Zcl{\Gen{S}}$ for some
    $S\subseteq\GL_n(\QQ)$. Let $H$ be a \done{closed} normal subgroup
    of~$G$ with finite index~$d$. Then $H=\Zcl{\Gen{S'}}$ where
    $S':=H \cap (S\cup S^{-1})^{\leq 2d+1}$.
\end{lemma}
\begin{proof}
  Define $H':=\Gen{S}\cap H$.  We claim that~$H'=\Gen{S'}$.
  \done{As $H$ was assumed to be closed,} it follows that
  \[\Zcl{\Gen{S'}}= \Zcl{H'}=\Zcl{\Gen{S}\cap H}=\Zcl{\Gen{S}}\cap H=H\;,\]
  establishing the lemma.

  It remains to establish the claim.  To this end, 
  observe that for $g,h \in \Gen{S}$ we have 
  $gh^{-1} \in H'$ if and only if $gh^{-1}\in H$, and hence  
  \[ |\Gen{S}/H'| \leq |G/H| = d \, .\]
  Then, since $H'$ is a subgroup of~$\Gen{S}$ of finite index at most~$d$,
  we can apply Schreier's Lemma~\cite[Lem.~4.\done{2}.1]{Seressbook} to obtain that 
  $H'$ is generated by the set 
  \[H' \cap\{t_1\, s\, t_2^{-1}\mid s\in S' \text{ and } t_1,t_2\in T\} \;, \]
  for any set~$T\subseteq \Gen{S}$ that contains a representative of
  each right coset of~$H'$ in~$\Gen{S}$.  But clearly we can take~$T$
  to be~$S^{\leq d}$.
\end{proof}

We are now equipped to
prove \cref{th:bound_finitely_generated_closure};
see \appref{app:main} for the detailed proofs of the claims.

\thbound

\begin{proof}
  Let $U:=\Zcl{\Gen{G_u}}$ be \done{the closure of the group} generated by the collection of unipotent
  elements in~$G$. 
%
%
  By \cref{lem:hrushovski_alggroup}, there  exists a \done{closed} normal subgroup $H\trianglelefteq G$ defined over~$\QQ$, such that $H/U$ is commutative, and
the index of $H$ in  $G$ is at most $J'(n)$, where $J'$ is as in the statement of \cref{lem:hrushovski_alggroup}.
  We will now use the characteristic properties of $U$ and $H$
  to obtain a
  bound on the degree of the polynomials defining~$G$. 
  We start by computing 
  a bound on the degree of the polynomials defining~$H$.
  

  Since $H$ is a \done{closed} normal subgroup of~$G = \Zcl{\Gen{S}}$ with index
  bounded by~$J'(n)$, it follows from \cref{lem:finite_index_subgroup_generators}
  that $H$ is the closure of the group generated by~$S' := H \cap (S\cup S^{-1})^{\leqslant 2J'(n)+1}$. Furthermore,
  since $S'$ is constructed as a finite multiplication of
  elements of~$S$, and their inverses, we can obtain the following
  bound.
  
  \begin{restatable}{claim}{clheightgeneratorsh} 	
    \label{claim:generators-h-height}
    The entries of the 
    matrices in $S'$ have height at most $$h':=(h^{n^3+n^2}n!n)^{2J'(n)+1}\;.$$	
  \end{restatable}
   
  Write $S'=\{g_1,\ldots,g_\ell\}\subseteq\GL_n(\QQ)$ where~$\ell \leq
  (2|S|)^{2J'(n)+1}$. Recall from~\cref{sub:cc} that for~$g\in\GL_n(\QQbar)$, we denote by~$g_s$ its semisimple part given by
  the unique decomposition $g=g_sg_u$.  Then we have
  \begin{align}
    \label{eq:equations-h}
   	H &= \Zcl{\Gen{S'}U} &&\text{$U \trianglelefteq H$ by \cref{lem:hrushovski_alggroup}} \nonumber \\
    &=  \Zcl{\Gen{g_1}\cdots\Gen{g_\ell}U} && \text{$H/U$ is commutative by \cref{lem:hrushovski_alggroup}\ref{quant1}}\nonumber\\
    &= \Zcl{\Gen{(g_1)_s)}\cdots\Gen{(g_\ell)_s}U}
        && \text{$gU=g_sU$ for all $g \in G$} \nonumber \\
    &= \Zcl{\Zcl{\Gen{(g_1)_s}} \cdots \Zcl{\Gen{(g_\ell)_s}}U} &&  
  \end{align}
    
  By the above, in order to compute a bound on the degree of the defining equations of~$H$, it suffices  to compute a bound on the degree of the equations that define~$U$, as well as a bound on the degree of equations that define each of the $\Zcl{\Gen{(g_i)_s}}$.
     
  Given a matrix~$g \in S'$, the
  eigenvalues~$\lambda_1,\ldots,\lambda_n$ of~$g$~and $g_s$ are the same. 
  \citeauthor{Mas88}'s Theorem~\cite{Mas88,CaiLZ00}, restated in \cref{th:masser} in~\cref{app:main}, provides a bound on the generators of the lattice of multiplicative relations of 
the form~$\lambda_1^{k_1}\cdots\lambda_n^{k_n}=1$ for integers
  $k_1,\dots, k_n\in \mathbb{Z}$.  By this bound and an argument similar to the one presented
  in \cite[Sec.~3.3]{DerksenJK05}, we obtain a degree bound for~$\Zcl{\Gen{g_s}}$.
   Define $f(n,h):=(cn^7n!\log(h))^n$, where $c$ is an absolute
  constant.
  \begin{restatable}{claim}{clonesemisimple}
    \label{claim:gen_gs_bound}
    $\Zcl{\Gen{g_s}}$ is bounded over $\QQ$ by $n\cdot f(n,h')$ for $g \in S'$.
  \end{restatable}

  By~\cref{claim:U-connected-bounded},  $U$ is bounded over~$\QQ$
  by $(n^3+1)^{2^{3n^2}}$.  Note that
  by \cref{eq:equations-h} we have a system of equations with $n^2+1$
  variables defining $H$, $\ell \cdot (n^2 + 1)$ variables defining
  the $\Zcl{\Gen{(g_i)_s}}$ and $n^2+1$ variables defining~$U$.
  We eliminate all variables from this system of equations except those defining~$H$. In order to compute 
  a bound on the degree increase that occurs in the elimination ideal of $H$, we use a
  result from~\citet*{Dube90}, stated earlier as \cref{prop:elm=ideal}.
  Setting $d:=\max\left(n\cdot f(n,h'),(n^3+1)^{2^{3n^2}}\right)$,
  we obtain the following bound.
  
  \begin{restatable}{claim}{clboundh}
    \label{claim:main-bound-h}
    $H$ is bounded over $\QQ$ by $(d+1)^{2^{(\ell+2)(n^2 + 1)}}$.
  \end{restatable}

  To conclude, it suffices to look at~$G$ as the union of
  all the cosets of~$H$ in~$G$. 
  That is, $G$ is the union of
  at most $J'(n)$ varieties, each of degree
  at most $(d+1)^{2^{(\ell+2)(n^2 + 1)}}$.
  Hence $G$ is bounded over~$\QQ$ by $
  (d+1)^{2^{(\ell+2)(n^2+1)J'(n)}}$.
  The final bound for~$G$ as stated in the theorem
  is obtained from the bounds for $d$, $\ell$, and~$J'$.
\end{proof}

\section{Chains of Linear Algebraic Groups}
\label{sec:chain}
 
In this section, we provide a quantitative version of the Noetherian
property of subgroups of~$\GL_n(\QQbar)$, i.e., a concrete bound on
the length of chains of algebraic subgroups.
%

\subsection{Chains of Algebraic Sets}\label{sub:chains}
Consider an algebraically closed field~$K$. Because the Zariski
topology is Noetherian, any strictly descending chain $X_0\supsetneq
X_1\supsetneq \cdots$ of algebraic sets in~$K^m$ must be
finite~\cite[Chap.~1]{HumphreysLAG}. Ascending chains of algebraic
sets can be infinite, but in the case of irreducible sets they are
finite and we have a stronger statement.  For irreducible algebraic
sets~$J$ and~$J'$, if $J\subseteq J'$ then
$\mathrm{dim}(J)\leq\mathrm{dim}(J')$ and if
$\mathrm{dim}(J)=\mathrm{dim}(J')$ then $J=J'$, where
$\mathrm{dim}(J)$ denotes the \emph{dimension}
of~$J$~\cite[Chap.~3]{HumphreysLAG}.  Thus, for an irreducible
algebraic set~$J$, any chain $J_0\subsetneq\cdots\subsetneq J_\ell$ of
irreducible algebraic subsets of~$J$ has length
$\ell\leq\mathrm{dim}(J)$. In particular, $\GL_n(K)$ has
dimension~$\mathrm{dim}(\GL_n(K))=n^2$.

This translates in the context of polynomials in $K[x_1,\dots,x_m]$
into, respectively, the ascending chain condition for ideals, and
uniform bounds on the length of chains of prime
ideals~\cite[Sec.~7.5]{BeckerW93}.  
Bounds on the length of
ascending chains of polynomial ideals yield constructive
versions of Hilbert's Basis Theorem, with applications for instance
to the ideal membership problem, and have been investigated for a
while~\cite{seidenberg72,socias92,mishra94,novikov99,aschenbrenner04}.
Note that, in order to bound the length of ascending chains $I_0\subsetneq
I_1\subsetneq\cdots$ of polynomial ideals, one must restrict the range
of allowed polynomials, e.g., by requiring the degree of the
polynomials in the basis for each~$I_{i}$ to be bounded by some function of
the index~$i$.  The resulting upper bounds are ackermannian in the
dimension \done{\cite[see, e.g.,][]{socias92,mishra94}}
, way beyond the elementary bounds of this paper.

\subsection{The Length of Chains of Linear Algebraic Groups}
We show that linear algebraic groups in dimension~$n$ over the
rationals behave more like irreducible algebraic sets than general
algebraic sets, in that there is a uniform bound on the length of
their chains (regardless of whether the chain is ascending or
descending) that depends only on~$n$.  We start
with \cref{lem:semisimple-chain}, which is restricted to the case in
which each group in the chain exclusively consists of semisimple
elements.

\begin{lemma} \label{lem:semisimple-chain}
    Let $n\in \NN$ and \done{$G_i=\Zcl{\Gen{S_i}}$ for
    $S_i\subseteq\GL_{n}(\QQ)$}, $1 \leq i \leq \ell$, be such that $G_1\subsetneq G_2\subsetneq\cdots \subsetneq G_\ell$.
    If each $G_i$ only consists of semisimple elements,  
    then $\ell \leq n^2 (2^{2n^2+3})!$.
\end{lemma}

\begin{proof} 	
  For each linear algebraic group $G_i$, denote by~$P_i$ the pistil
  of~$G_i$, that is, the intersection of all maximal tori that
  contain~$\IdComp{(G_i)}$. Since the sequence
  $\IdComp{(G_1)}\subseteq \IdComp{(G_2)} \subseteq \cdots \subseteq \IdComp{(G_\ell)}$
  is increasing, $P_1 \subseteq P_2 \subseteq \cdots \subseteq
  P_{\ell}$ is also an increasing sequence.

  By definition, each $\IdComp{(G_i)}$ is irreducible, thus there are
  at most $n^2$ distinct \done{possibilities} for the $\IdComp{(G_i)}$, and hence
  also for the $P_i$. Furthermore, as shown
  in \cref{claim:pistil-bounded}, each pistil~$P_i$ is bounded by~$1$.
    Since $G_i \cap P_i$ is a normal subgroup of $G_i$ for all $i$
by \cref{lem:hrushovski_semisimple}\ref{pistil2}, we can
  use \cref{th:homomorphism_quotient} to show that each quotient
  $G_i/(G_i \cap P_i)$ is isomorphic to a finite linear group in
  dimension $m:=2^{2(n^2+1)}$.
Furthermore, \cref{lem:hrushovski_semisimple}\ref{pistil2} also shows that
 the order of each quotient
  $G_i/(G_i \cap P_i)$ will be at most~$(2m)!$.  Taking the product of
  the number of distinct~$P_i$ and the number of distinct finite
  quotients~$G_i/(G_i\cap P_i)$ we obtain the upper bound~$n^2 (2m)!$
  for the length~$\ell$ of the chain.
\end{proof}
 
Note that~\cref{lem:semisimple-chain} immediately applies to the case
where the groups $G_1, \dots, G_{\ell}$ are \done{closures of groups generated by unitary
matrices}.  For the general case, we provide the following bound.

\begin{theorem}
    Let $n\in \NN$ and let $G_i=\Zcl{\Gen{S_i}}$ for $S_i\subseteq \GL_{n}(\QQ)$, $1 \leq
    i \leq \ell$, be such that
    $G_1\subsetneq G_2\subsetneq\cdots \subsetneq G_\ell$. 
    Then we have~$\ell \leq \exp^6(\mathrm{poly}(n))$.
\label{theo:chain-all}
\end{theorem}

\begin{proof}
  For each linear algebraic group $G_i$, let $U_i=\Zcl{\Gen{(G_i)_u}}$
  be the \done{closure of the} subgroup generated by the unipotent elements in~$G_i$.
  By \cref{claim:U-connected-bounded}, each $U_i$ is $d$-bounded
  over~$\QQ$
  for $d:=(n^3+1)^{2^{3n^2}}$.
  Consequently,
  by \cref{th:homomorphism_quotient} each quotient~$G_i/U_i$ is
  isomorphic to a linear algebraic group in dimension
  $p:=2^{2(n^2 + d)^d}$, which moreover consists exclusively of
  semisimple elements.
        
  Observe that the sequence $U_0 \subseteq
  U_1 \subseteq \cdots \subseteq U_{\ell}$ is increasing.  Since
  by~\cref{claim:U-connected-bounded} the~$U_i$ are irreducible, we
  can decompose this sequence into at most~$n^2$ blocks where the
  dimension of the~$U_i$ is constant in each block but increases from
  one block to the next.  Within each block, say from index~$k$ to~$r$
  where $1\leq k\leq r\leq \ell$, since the dimension is constant,
  $U:=U_k=U_{k+1}=\cdots=U_r$.
   The sequence $(G_k/U) \subseteq \cdots \subseteq  (G_{r}/U)$ 
  is strictly increasing and hence 
  by~\cref{lem:semisimple-chain} has
  length at most~$p^2(2^{2p^2+3})!$.
  Taking the product of the number of blocks and the length of each
  block we obtain the
  upper bound~$n^2 p^2(2^{2p^2+3})!$
  for the
  length of the chain~$\ell$.
\end{proof}

To conclude this section, recall that a number field~$k$ is a finite
extension of the field of rational numbers~$\QQ$.  As will be briefly
discussed in \cref{sub:field}, the order of a finite subgroup of
$\GL_m(k)$ is at most~$(2m[k:\QQ])!$ where~$[k:\QQ]$ is the dimension
of~$k$ as a vector space over~$\QQ$.  This allows to
extend~\cref{lem:semisimple-chain,theo:chain-all} to a chain of linear
algebraic groups generated over a number field, with the appropriate
changes to the upper bound on the length~$\ell$ of chain.

\begin{corollary}
  Let $n\in \mathbb{N}$, $k$ be a number field,
  and $G_i=\Zcl{\Gen{S_i}}$ for $S_i\subseteq  \GL_{n}(k)$, $1 \leq
  i \leq \ell$, be such that
  $G_1\subsetneq G_2\subsetneq\cdots \subsetneq G_\ell$.  Then
  $\ell\leq\mathrm{exp}^1\big(\mathrm{poly}([k:\QQ])\,\mathrm{exp}^5(\mathrm{poly}(n))\big)$,
  and $\ell\leq\mathrm{exp}^1\big(\mathrm{poly}([k:\QQ])\,\mathrm{exp}^1(\mathrm{poly}(n))\big)$ if each~$G_i$
  consists of semisimple elements.
\end{corollary}

\section{Discussion}\label{sec:discussion}

\subsection{Extension to Number Fields}\label{sub:field}
In this paper, we have focused on the case of rational matrices, but
the results can be extended to matrices over a number field~$k$ at
minimal additional expense.  Upon inspection of our proofs, it turns
out that the only place where working in this more general setting has
an impact is on the order of finite subgroups of $\GL_p(\QQ)$ in
\cref{theo-order-finite-G}.
Over a finite extension~$k$
of~$\QQ$, the bound needs to be updated
to~$(2p[k:\QQ])!$ in order to take the dimension of~$k$ as a vector
space over~$\QQ$ into account.  The remainder of the bounds in the
paper have to be adapted mutatis
  mutandis.

\subsection{Dependence on the Dimension and Height}
Our upper bound in \cref{th:bound_finitely_generated_closure} on
the degree of the polynomials defining the Zariski closure of a
finitely generated matrix group is a function of both the dimension of
the generators and the height of the entries of the generators.  The
following examples illustrate that such a degree bound
necessarily depends on both parameters.
\begin{example}[Degree depends on the dimension]\label{ex:dim}
  Recall that the special linear group \MB{$\SL_n(\ZZ):=\{ M \in
    \GL_n(\ZZ) : \det(M)=1\}$ is generated by
  three matrices with entries in $\{-1,0,1\}$
  \cite[p.~35]{MacDuffee1933}, i.e., with maximum height $h=1$.  The
  Zariski closure of $\SL_n(\ZZ)$ is $\SL_n(\QQbar)$.\footnote{This
    follows e.g.\ from the Borel Density Theorem~\cite[Sec. 4.5 and
    Sec. 7.0]{Morris01}, but can also be established directly by an
    elementary argument.}}
 We observe that \emph{any} non-zero polynomial $f$ that vanishes on $\SL_n(\QQbar)$ has (total) degree at 
least $n$.  Indeed, evaluating $f$ on the diagonal
matrix $\mathrm{diag}(x,\ldots,x)$, for some variable $x$, yields a univariate polynomial $g$ 
that vanishes on every $n$-th root of unity.  It
follows that $x^n-1$ divides $g$ and so $g$, and hence also $f$, has
degree at least $n$. 
\end{example}
\begin{example}[Degree depends on the height]\label{ex:height}
  Consider the case of a single $2\times 2$ matrix $g:=\diag(2^p,1/2)$ 
  for some $p\in\NN$, i.e., with height $h=2^p$.
  By \cite[Lem.~6]{DerksenJK05}, the vanishing ideal of~$\Zcl{\Gen{g}}$ is generated by 
  the multiplicative relations among the diagonal elements of~$g$. These relations have one generator, namely~$(2^p)^1(1/2)^p=1$.
  Any non-zero polynomial~$f$ that vanishes on $\Zcl{\Gen{g}}$ must also vanish on 
  $\{\diag(x,y):xy^p=1\}$, and thus be of degree at least~$1+p\geq\log h$.  
\end{example}

\subsection{Related Work on the Algorithm of \citeauthor*{DerksenJK05}}
\label{sub-DJK}
Given a \done{finite} set of rational matrices, an algorithm computing the
\done{Zariski closure~$G$ of their generated group} is presented in~\cite{DerksenJK05}. In a nutshell, the algorithm computes better and better under-approximations~$H$ of the identity component $\IdComp{G}$ until the quotient $G/H$ is finite.
The same principle is used in an algorithm described in~\cite[Sec.~4.6]{deGraafBook}, with part of the computation performed at the level of the Lie algebra of~$G$. 
%
%
%
	One obstacle to analysing the complexity of these algorithms is the following.
	In order to obtain a quantitative statement about $H$, it seems
        necessary to understand which elements will decrease the index
        of $H$ in $G$ if they are added to $H$. This, in turn,
        crucially relies \done{not only} on the degree \done{(as
          analysed in \cref{th:homomorphism_quotient}), but also
          on the} height of
	the coefficients of the quotient map $\phi_H$~\citep{nosanM1}. 
	Bounding these coefficients requires a delicate analysis of
the construction of~$\phi_H$. \done{See \appref{app:derksen} for more details.}

\subsection{Related Work on \citeauthor{HR}'s Algorithm}
\Citeauthor*{HR}'s work has been subsequently pursued by \citet*{Feng}
and \citet*{Sun}, leading to a triple exponential algorithm for
computing the Galois group of a linear differential equation.  A
different approach was recently pursued by \citeauthor*{amzallag2019},
who introduced the notion of toric envelope to obtain a single
exponential bound for the first step in \citeauthor*{Feng}'s
formulation of \citeauthor{HR}'s algorithm,
which is qualitatively
optimal~\cite{amzallag2019,amzallagPhD}.

\SkipTocEntry\section*{Acknowledgements}
We thank Ehud Hrushovski for an explanation of~\cite{HR} and other
  useful suggestions, and the anonymous reviewers for their
  comments and corrections.  A. Pouly is supported by the European Research Council (ERC) under the 
  European Union's Horizon 2020 research and innovation programme (Grant 
  agreement No. 852769, ARiAT).

\clearpage
\sloppy
\printbibliography
\clearpage
\markboth{APPENDICES}{APPENDICES}\clearpage
\begin{appendix}
 
\section{Folklore Lemma}\label{app:prelim}

\boundkgroup*

\begin{proof}[Proof of~\ref{itm:k_bounded}]
We may assume without loss of generality that $X$ is the zero set of a \emph{finite} collection 
of polynomials 
$I \subseteq L[\boldsymbol{x}]$ of degree at most $d$ over a tuple
$\boldsymbol{x}$ of variables.
Let $F\subseteq L$ be the normal closure of the field extension of $k$ generated by the 
coefficients of the polynomials in~$I$.  Then $F$ is a finite Galois
extension of~$k$. Given a Galois extension $F/k$, we call
the group of automorphims of~$F$ that fix~$k$ a \emph{Galois group} and denote it by $\Gal(F/k)$.

Given $\alpha \in F$, write $\Tr(\alpha):=\sum_{\sigma\in\Gal(F/k)}\sigma(\alpha)$ for the trace of $\alpha$ relative to $F/k$.
We extend $\Tr$ coefficient-wise to polynomials $p\in F[\boldsymbol{x}]$.    
    Since $F/k$ is a Galois extension of~$k$,   we have $\Tr(p)\in k[\boldsymbol{x}]$ for $p\in F[\boldsymbol{x}]$.
    Let $b_1,\ldots,b_n$ be a basis of $F$ over $k$ and consider the matrix $M=(\Tr(b_ib_j))_{ij}$ such that $\det(M)^2 \neq 0$ is the discriminant of the basis.  Given $\alpha=\sum_{i=1}^n\alpha_ib_i\in F$, where $\alpha_1,\ldots,\alpha_n \in k$, for all $j\in\{1,\ldots,n\}$
    we have $\Tr(\alpha b_j)=\sum_{i=1}^n\alpha_i\Tr(b_ib_j)$.
    Writing $M^{-1}=(c_{ij})$, 
    it follows that for all $i\in\{1,\ldots,n\}$ we have
    $\alpha_i = \sum_j c_{ij} \Tr(ab_j)$.  We conclude that 
    $\alpha=\sum_{i,j=1}^n \Tr(\alpha b_j)c_{ij}b_i$. We then have that
    for any polynomial $p\in F[\boldsymbol{x}]$, $p=\sum_{i,j=1}^n\Tr(p b_j)c_{ij}b_i$, that is,
    $p$ is a $F$-linear combination of the polynomials $\Tr(p b_1),\ldots,$ $\Tr(p b_n)$ in $k[\boldsymbol{x}]$.
    
    
    Suppose now that $p \in F[\boldsymbol{x}]$ is a polynomial that vanishes on $X$.  Since each $\sigma \in \Gal(F/k)$ extends to a map 
    in $\mathrm{Aut}(L/k)$, we have $\sigma(X)=X$ for all
    $\sigma \in \Gal(F/k)$.  Thus
    each polynomial $\Tr(pb_i)$, for $i=1,\ldots,n$, also vanishes on $X$.
    We conclude that $X$ is also the zero set of the collection
    $I':=\{\Tr(b_ip):i=1,\ldots,n,p \in I\} \subseteq k[\boldsymbol{x}]$.  But the polynomials in $I'$ all have degree at most $d$.
\end{proof}

\begin{proof}[Proof of~\ref{itm:dense}]
    We claim that $X$ is stable under every automorphism~$\sigma\in
    \Aut(L/k)$, which will prove that $X$ is $d$-bounded over~$k$ by~\ref{itm:k_bounded}.
To see this, let the polynomial $p\in L[\boldsymbol{x}]$ be a
    polynomial that vanishes on $X$ and $\sigma\in
    \Aut(L/k)$.
Then $p^\sigma \in L[\boldsymbol{x}]$, defined by applying $\sigma$ point-wise to the coefficient of~$p$, will vanish on~$\sigma(X)$. But $\sigma(X)$ contains
$\sigma(X\cap k^m)=X\cap k^m$, which is a dense subset of~$X$. Hence $p^\sigma$ vanishes on 
$X$, equivalently $p$ vanishes on $\sigma^{-1}(X)$. We conclude that $X$ is stable under $\sigma^{-1}$. Since $\sigma$ was arbitrary, $X$ is stable under $\Aut(L/k)$.   
\end{proof}






\section{Computations of Degree Bounds}
\label{app:main}

\clheightgeneratorsh*

\begin{proof}
To compute a height bound on the entries of $S'$, it suffices to 
compute a height bound~$\eta$ on the entries of matrices in~$S$ and
their inverses. Then
by definition of~$S'$, the height of its elements will be bounded 
by the $(2J'(n)+1)$th power of~$n\eta$.

Recall that \done{the matrices in~$S$ have their} entries bounded by~$h$. 
Fix $g\in S$ and let $\ell\leq h^{n^2}$ be the
least common multiple of the denominators of the entries of~$g$.
Recall that the~$e_k$ are the standard basis, where~$e_k$ denotes the
vector with a~$1$ in the~$k$th coordinate and~$0$'s elsewhere.
The $k$th column of~$g^{-1}$ can be computed by
setting up~$gX=e_k$. 
For simplicity of computation, we multiply both sides with the $n\times n$ diagonal matrix $\mathrm{diag}(\ell,\ldots,\ell)$.
By Cramer's rule and an approximation on the
magnitude of the determinant of~$g$, we compute the
bound~$h^{n^3+n^2}n!$ on the height of entries in~$g^{-1}$. The claim
immediately follows.   
\end{proof}

The following results are needed for the proof
of \cref{claim:gen_gs_bound}. Recall that a number is \emph{algebraic}
if it is the root of a non-zero polynomial in $\QQ[x]$.  For
$\alpha\in \QQbar$, the \emph{defining polynomial} of $\alpha$ is the
primitive polynomial in $\mathbb{Z}[x]$ of minimal degree having
$\alpha$ as a root.  The \emph{height} of $\alpha$ is the maximum
absolute value of the coefficients of its defining polynomial.  Note
that this yields the same notion of height as in the remainder of the
paper: for a rational number $\frac{a}{b}$ with $a,b$ coprime
integers, the defining polynomial is $bx-a$ and the height is
$\max(|a|,|b|)$.

We start by proving a bound on the generators of a lattice of
multiplicative relations of the
form~$\lambda_1^{k_1}\cdots\lambda_n^{k_n}=1$ for the algebraic
numbers~$\lambda_1,\ldots,\lambda_n$ and integers $k_1,\ldots,
k_n\in \mathbb{Z}$. Such relations form an additive subgroup
of~$\mathbb{Z}^n$; see~\cite{Mas88,CaiLZ00}.

\begin{theorem}[Masser]\label{th:masser}
  Fix any $n\geq 1$.  Let $\lambda_1,\ldots,\lambda_n$ be non-zero
  algebraic numbers having height at most $h$ over $\QQ$ and let
  $D:=[\QQ(\lambda_1,\ldots,\lambda_n):\QQ]$.  Then the group of
  multiplicative relations
  \begin{gather}
    L:=\{(k_1,\ldots,k_n)\in \mathbb{Z}^n : \lambda_1^{k_1}\cdots \lambda_n^{k_n} = 1\}
    \label{eq:lattice}
  \end{gather} 
  is generated (as a subgroup of~$\ZZ^n$) by a collection of vectors
  all of whose entries have absolute value at most
  \[ (cn\log(h))^{n-1}D^{n-1}\frac{(\log(D+2))^{3n-3}}{(\log\log(D+2))^{3n-4}} \, ,\]
  for some absolute constant~$c$.
  \label{the:masser}
\end{theorem}

\begin{corollary}
Let $g \in \GL_n(\QQ)$ have entries of height at most~$h$.   Writing
$\lambda_1,\ldots,\lambda_n$ for the eigenvalues of~$g$, the
lattice~$L$ of multiplicative relations in \eqref{eq:lattice} is
generated by a collection of vectors, all of whose entries having
absolute value at most~$f(n,h):=(cn^7n!\log(h))^n$ for some absolute
constant~$c$.
\label{cor:lattice-height}
\end{corollary}
\begin{proof}
Let $\ell \leq h^{n^2}$ be the least common multiple of the denominators of the entries of~$g$.  The eigenvalues of~\done{$g\ell$} are roots of the polynomial \[P(x):=\mathrm{det}(x\GId-g\ell)=\sum_{\sigma \in S_n} \mathrm{sgn}(\sigma) \prod_{i=1}^n (x\GId-g\ell)_{i,\sigma(i)} \, .\] 
The coefficients of this polynomial are integers of absolute value at most~$H:=2^nh^{n^3}n!$. Hence $\lambda_1,\ldots,\lambda_n$ have height at most~$H$.

Applying~\cref{the:masser}, since 
the field $\QQ(\lambda_1,\ldots,\lambda_n)$ has degree at most~$D:=n!$ over~$\QQ$, the lattice~$L$ of multiplicative relations in~\eqref{eq:lattice} is generated by vectors whose entries have absolute value at most $(cn^7n!\log(h))^n$ for some absolute constant~$c$.
\end{proof}

The next proposition gives a degree bound for the Zariski-closure of
the group generated \done{by} a single semisimple matrix.

\begin{restatable}{proposition}{proponematrix}
  \label{prop:gen_gs_height}
  Let $n\in \NN$ and $g \in \GL_n(\QQ)$ have entries of height at
  most~$h$. Then $\Zcl{\Gen{g_s}}$ is bounded by $n\cdot f(n,h)$ over
  $\QQ$, where $f(n,h):=(cn^7n!\log(h))^n$ for some absolute
  constant~$c$.
\end{restatable}
\begin{proof}
  Let $g_s$ be written in the form $g_s=h^{-1} a h$, where 
  $a=\mathrm{diag}(\lambda_1,\ldots,\lambda_n)$ is diagonal.
  Write 
  \[ L:=\{ (k_1,\ldots,k_n) \in \mathbb{Z}^n : \lambda_1^{k_1} \cdots \lambda_n^{k_n}=1\}\]
  for the group of multiplicative relations among
  $\lambda_1,\ldots,\lambda_n$.   It is known (see,
  e.g., \cite{DerksenJK05}) that
  \[ \Zcl{\Gen{a}}= \{ \mathrm{diag}(x_1,\ldots,x_n) :
  x_1^{k_1} \cdots x_n^{k_n}=1 \text{ for all } (k_1,\ldots,k_n) \in L\} \; . \]
  
  But by \cref{cor:lattice-height}, $L$ is generated as an additive
  subgroup of~$\ZZ^n$ by vectors $(k_1,\ldots,k_n) \in \ZZ^n$ where
  the~$k_i$ have absolute value at most~$f(n,h)$.  It follows
  that~$\Zcl{\Gen{a}}$ is bounded over $\QQ$ by~$nf(n,h)$.  But then
  we have that $\Zcl{\Gen{g_s}}=h^{-1}\Zcl{\Gen{a}}h$ is also bounded
  by~$nf(n,h)$ over~$\QQ$.
\end{proof}

\clonesemisimple*

\begin{proof}
  Recall that by \cref{claim:generators-h-height}, all matrices~$g\in
  S'$ have height bounded by~$h'$, where $h'$ is as given in the
  statement of \cref{claim:generators-h-height}. The proof of the
  claim then follows immediately from \cref{prop:gen_gs_height}.
\end{proof}

\clboundh*
    
\begin{proof}
  By \cref{lem:subgroup_generated_by_unipotents}, $U$ is bounded
  over~$\QQ$ by~$(n^3+1)^{2^{3n^2}}$, and by \cref{claim:gen_gs_bound},
  each group $\Zcl{\Gen{(g_i)_s}}$ is bounded over~$\QQ$ by~$n\cdot f(n,h')$
  where~$f$ is as in the statement
  of~\cref{cor:lattice-height}.
  
   Let $d:=\max\left(n\cdot f(n,h'),(n^3+1)^{2^{3n^2}}\right)$.     
  By~\cref{prop:elm=ideal}, $H$ is bounded over~$\QQ$
  by~$(d+1)^{2^{(\ell+2)(n^2+1)}}$.
\end{proof}

 \section{Analysis of the algorithm of Derksen, Jeandel, and Koiran}
 \label{app:derksen}
 \label{app-DJK}
In this section, we give a quick overview of the algorithm
of \citet*{DerksenJK05}.  We then explain why analysing its
complexity requires a delicate analysis that is unlikely to yield a
good upper bound.  We refer to~\citep{nosanM1} for a more in-depth
analysis.

\subsection{Overview}
In a nutshell, the algorithm builds the closure
$G:=\Zcl{\Gen{g_1,\ldots,g_N}}$ of a finite set
$\{g_1,\ldots,g_N\}\subseteq \GL_n(\QQ)$ by computing successive
under-approximations $H\cdot S$ where~$H$ is an irreducible group
and~$S$ a finite set, until $G=H\cdot S$.
Initially, $H:=\{\GId\}$ and $S:=\{\GId,g_1,\ldots,g_N\}$.  The
algorithm performs a main loop and at each iteration, one of the
following will happen:
\begin{enumerate}
    \item $S$ increases in a way that makes $H\cdot S$ larger, or
    \item $H$ increases, or
    \item the algorithm stops.
\end{enumerate}
We do not repeat the details of the algorithm nor its proof of
correctness.  It suffices to know that~$H$ is an irreducible normal subgroup of~$G$
at the beginning of every iteration.\footnote{This is technically not
true in the algorithm as written in~\cite{DerksenJK05}, but a minor
modification ensures it without changing the complexity~\cite{nosanM1}.}  This
implies that $G/H$ is an algebraic group and, therefore, only two
possibilities exist:
\begin{itemize}
\item If $G/H$ is finite then~$H$ will not grow anymore and the
  algorithm essentially ``fills up'' the quotient~$G/H$ by adding
  elements to~$S$, which takes finitely many steps.
\item If $G/H$ is infinite then the algorithm will eventually find an
  element~$h$ of infinite order in~$G/H$ and replace~$H$ by
  $\Zcl{H\cdot\IdComp{\Zcl{\Gen h}}}$ and then close~$H$ under various
  operations to ensure it is a still a normal subgroup of~$G$.
\end{itemize}
By virtue of being an irreducible variety, $H$ can only increase~$n^2$
times, by the dimension argument recalled in \cref{sub:chains}.
Hence, the delicate part of the analysis is to understand how long it
takes for the algorithm to either discover an element of infinite
order to make $H$ grow, or to ``fill up'' $G/H$ when it is finite.

\subsection{Case of a Finite Quotient}
Consider the latter situation: the quotient~$G/H$ is finite.  In this
case, the algorithm needs to discover all the elements of~$G/H$ by
taking products of elements in~$S$.  Even if we
disregard the fact that not all products will produce new elements, it
might still take~$|G/H|$ steps for the algorithm to finish.  Hence, we
need to obtain a bound on the size of~$G/H$.

Now, $G/H$ being an algebraic group, it can be embedded as a matrix
group in some dimension~$d$.  It can be shown that, if the generators
$g_1,\ldots,g_N$ have rational entries, then in fact~$G/H$ can be
embedded in $\GL_d(\QQ)$ using the quotient map~$\phi_H$.  It then
follows by \cref{theo-order-finite-G} that the size of~$G/H$ is
bounded by a function of~$d$ only ($(2d)!$ is an upper bound).  In
other words, we need to obtain a bound on~$d$.  Such a bound is shown
by \citet{Feng} and recalled in~\cref{th:homomorphism_quotient}, and
depends on the degree of the polynomials that define~$H$ (and this
bound cannot be improved in general).\footnote{There is not
necessarily a unique way of choosing polynomial equations for a given
variety, but one could consider the smallest degree among all such
equations.} In summary: we need to track the degree of (some)
polynomial equations that define~$H$.

\subsection{Case of an Infinite Quotient}
In order to track the degree of the equations defining~$H$, we need to
study the case where~$H$ grows.  Recall that this happens when the
algorithm finds an element $h$ of infinite order in $G/H$, and then
$H$ becomes $\Zcl{H\cdot\IdComp{\Zcl{\Gen h}}}$. This, in turn,
requires to bound the degree of equations defining $\IdComp{\Zcl{\Gen
h}}$.  It can be seen that $\IdComp{\Zcl{\Gen h}}$ essentially encodes
the multiplicative relations between the eigenvalues of~$h$.  The only
bound on the degree of those relations that we are aware of is given
by \citeauthor{Mas88}'s Theorem~\citep{Mas88}
(see \cref{th:masser}).\footnote{We note in passing that the
complexity of finding those relations is a nontrivial topic in
itself. The naive algorithm that enumerates all relations up to the
bound and checks them one by one has a terrible complexity.  A
potentially more efficient algorithm is given by \citet*{Ge} but it
does not improve the bound on the degree.}  The crucial issue,
however, is that the \emph{degree bound} given by \citeauthor{Mas88}'s
Theorem depends on the \emph{height of the coefficients}.  In summary:
we need to understand the height of elements~$h$ of infinite order in
$G/H$.

The algorithm needs to find an element~$h$ of infinite order in $G/H$
to make progress.  The only approach that we are aware of to perform
this task reduces the group~$G/H$ modulo some well-chosen prime~$p$,
and is used in the current state-of-the-art computational group theory
algorithms \citep[e.g.,][]{detinko13}.  Indeed, a well-known lemma
by \citet{minkowski87} shows that the kernel of the reduction map is
torsion-free and therefore that it suffices to find two elements
in~$G/H$ with the same reduction to obtain a matrix in the kernel that
will then have infinite order.  One can then use the fact that the
image of the group through the reduction map has size at
most~$p^{n^2}$ to argue that we only need to look at~$p^{n^2}$
products to find an element of infinite order.  This means that, in
order to bound the height of the element~$h$ of infinite order, we
need a bound on~$p$.

For this reduction to work, we need~$p$ not to divide any denominator
of the coefficients of the generators of~$G/H$.  We are led, again, to
the study of the quotient map in \cref{th:homomorphism_quotient} since
the generators of $G/H$ are $\phi_H(g_1),\ldots,\phi_H(g_N)$.
Unfortunately, the height of $\phi_H(g_i)$ depends on the height of
the coefficients of $\phi_H$, which is not analysed in
\cref{th:homomorphism_quotient}.  In fact, we are not aware of any
result in the literature that provides a bound on this height.  A good
foundation to obtain such a bound is described in~\cite{nosanM1} and
could, in principle, be obtained by a combination of a careful
analysis of the quotient, the use of quantifier elimination, and
bounds on the degree of bases of polynomial ideals.  Unfortunately, it
is clear that it would in turn depend in a nontrivial way on the
degree and height of the equations that define~$H$. Furthermore, as
conceived above, it would incur at least an exponential blow up.

\subsection{Conclusion}
To summarise, a complete analysis of the algorithm
of \citet*{DerksenJK05} seems to require to:
\begin{itemize}
\item analyse of the height of the quotient map
  in \cref{th:homomorphism_quotient}, and
\item track the degree and height of the equations that define~$H$.
\end{itemize}
Furthermore, we note again that both quantities (degree and height)
potentially suffer from \emph{at least} an exponential blow-up each
time~$H$ increases.  As~$H$ increases up to $n^2$ times, this means
that the obtained bound would be at least a tower of~$n^2$
exponentials, i.e., a non-elementary complexity bound.  In conclusion,
it is not clear what such a bound would be and the technical obstacles
to obtain it are significant.

\end{appendix}

\end{document}